\documentclass[aps,prl,reprint,amssymb]{revtex4-1} 
\usepackage{amsmath}
\usepackage[amsmath, thmmarks]{ntheorem}
\usepackage{enumerate}
\usepackage{bbm}
\usepackage{hyperref}
\usepackage{tikz}
\usetikzlibrary{fit, arrows, positioning, shapes} 

\theorembodyfont{\normalfont}
\newtheorem{dfn}{Definition}

\theoremsymbol{\ensuremath{\scriptstyle\blacksquare}}
\newtheorem{example}[dfn]{Example} 
\theoremstyle{nonumberplain}
\theoremheaderfont{\itshape}
\newtheorem{proof}{Proof}

\theoremclass{LaTeX}
\newtheorem{thm}[dfn]{Theorem} 
\newtheorem{theorem}[dfn]{Theorem} 
\newtheorem{proposition}[dfn]{Proposition}
\newtheorem{lemma}[dfn]{Lemma}
\newtheorem{cor}[dfn]{Corollary}
\newtheorem{obs}[dfn]{Observation}

\def\scaled{\let\onleft=\left\let\onright=\right}
\def\unscale{\let\onleft=\relax\let\onright=\relax}
\unscale
\newcommand{\bbone}{\ensuremath{\mathbbm 1}}

\newcommand{\defeq}{\ensuremath{:=}}
\newcommand{\setdef}{\ensuremath{\;\vert\;}}

\newcommand{\ocmpl}{\ensuremath{^{\mathrm{\perp}}}}

\def\scaled{\let\onleft=\left\let\onright=\right}
\def\unscale{\let\onleft=\relax\let\onright=\relax}

\newcommand{\set}[1]{\ensuremath{\onleft\{ #1\onright\}}\unscale}

\newcommand{\reals}{\ensuremath{\mathbb R}}

\newcommand{\borel}{\ensuremath{\mathrm{B}}}

\newcommand{\op}{\oplus}
\newcommand{\notperp}{\mathbin{\perp\kern-.9em/}}
\newcommand{\compat}{\mathbin{\leftrightarrow}}
\newcommand{\notsim}{\mathbin{\sim\kern-.9em/}}

\def\quotient#1#2{%
    \raise1ex\hbox{$#1$}\Big/\lower1ex\hbox{$#2$}%
}

\newcommand{\true}{\ensuremath{\texttt{true}}}
\newcommand{\false}{\ensuremath{\texttt{false}}}

\begin{document}

\title{Non-signalling theories and generalized probability}
\author{Tomasz I. Tylec}
\email{tylec@cft.edu.pl}
\author{Marek Ku\'s}
\affiliation{Center for Theoretical Physics, 
    Polish Academy of Sciences, 
    Aleja Lotnik\'ow 32/46, 
    02-668 Warsaw, Poland}
\author{Jacek Krajczok}
\affiliation{Faculty of Physics,
   University of Warsaw,
   Pasteura 5,
   02-093 Warsaw, Poland}

\begin{abstract}
    We provide mathematicaly rigorous justification
    of using term \emph{probability} in connection to
    the so called non-signalling theories, 
    known also as Popescu's and Rohrlich's box worlds.
    No only do we prove correctness of these models 
    (in the sense that they describe composite system
    of two independent subsystems)
    but we obtain new properties of non-signalling boxes
    and expose new tools for further investigation.
    Moreover, it allows strightforward generalization
    to more complicated systems.
\end{abstract}

\maketitle

\section{Introduction}

The idea of so-called \emph{non-signalling} theories,
or \emph{box worlds},
that originates from the paper of Popescu and Rohrlich \cite{popescu1994quantum},
became a popular tool in certain areas related to quantum information theory,
like proving security of ciphering protocols \cite{PhysRevLett.97.120405}
or finding bounds on communication complexity \cite{PhysRevLett.96.250401}
(to name only a few references out of many).
By a \emph{box world} one usually mean a physical model
\cite{barrett2007information}  % TODO check this ref
build from two ``black boxes'', 
where each box is characterized by a finite set of input values
(that correspond to selection of an ``observable'')
and, for each input, finite set of output values.
To fix notation, let us denote particular box world
by $(\mathcal U, \mathcal V)$, where
$\mathcal U = \set{\mathcal U_1, \dots, \mathcal U_N}$,
$\mathcal V = \set{\mathcal V_1, \dots, \mathcal V_M}$ and
$\mathcal U_i,\mathcal V_j$ are sets of outcomes
for $i$-th input on left box and $j$-th input on right box,
respectively.

A state of a box world is described with the help 
of the following thought experiment:
we supply boxes with a large stream of input values 
and write down obtained outputs.
Then we compute frequency $P(\alpha\beta| ab)$ of getting outputs
$\alpha, \beta$ given the inputs $a, b$ on the left and right box respectively
\cite{barrett2007information}.

Obviously, $P$ must be a non-negative function and
\begin{equation*}
    \sum_{\alpha\in\mathcal U_a, \beta\in\mathcal V_b} P(\alpha \beta | ab) = 1,
    \qquad a=1,\dots,N, b=1,\dots M.
\end{equation*}
Moreover, we want boxes to be in some sense independent,
thus we impose additional restrictions on $P$, 
namely the \emph{non-signalling} conditions
\begin{align}
    \sum_{\alpha'\in\mathcal U_a} P(\alpha' \beta| a b) &= 
    \sum_{\alpha'\in\mathcal U_c} P(\alpha' \beta| c b) \nonumber\\
    \sum_{\beta'\in\mathcal V_b} P(\alpha \beta'| a b) &= 
    \sum_{\beta'\in\mathcal V_c} P(\alpha \beta'| a c)
    \label{eq:nonsignal}                                                     
\end{align}
satisfied for any $a, b, c, \alpha, \beta$. 
These express the intuitive idea that the output of one of the boxes
does not depend on the input of the other. 
It also reflects Einstein's causality principle in the model,
as boxes could be placed in space-like separated regions of space-time.
It is assumed that any $P$ that satisfies all these properties
(non-negativity, normalization and non-signalling) is an admissible state
for a given box world.
In the sequel we will call such $P$ a \emph{PR-state} of the box world.

The problem with the above widely used formulation lies in the fact, 
that $P$ is interpreted as a \emph{probability} of getting output
$\alpha\beta$ given input $ab$.
This interpretation is based on the thought experiment we discussed previously.
We argue that applying frequentist definition of probability 
to thought experiments is unjustified. 
The number of ``paradoxes'' related to probability,
like Bertrand's paradox, Loschmidt's paradox or Monty Hall problem
to name only a few, shows that we should not rely solely on intuition
when we talk about probability.
On the other, hand without sound probabilistic interpretation of $P$
the meaning of non-signalling conditions is questionable
and so the whole idea of non-signalling theories.

Lacking of any physical realization of box worlds,
we are forced to justify probabilistic interpretation of $P$
on the basis of axiomatic approach to probability.
Because box worlds obey neither classical nor quantum probability rules
(cf.\ \cite{barrett2005nonlocal}) it is clear that we need framework
that generalizes standard Kolmogorov's axioms of probability.
Out of the two widely used such frameworks:
operator algebra approach (that focused on the algebra of random variables)
and quantum logic approach (that focused on a partialy ordered set of random events),
only the latter is capable of ``super-quantum'' generalization.

The paper is organized in the following way.
We begin with a very brief introduction to quantum logics
(detailed exposition can be found e.g.\ in \cite{pulmanova2007}).
Then we construct the logic of an arbitrary box world $(\mathcal U, \mathcal V)$.
We prove that in the framework of quantum logic
it indeed can describe system composed of two independent subsystems
and discuss some of its properties.
We end the paper with some remarks about generalization
of box worlds to larger number of boxes.

Finally we would like to mention that
there is another mathematically rigorous attempt to formalize box world theories
called General Probability Theory 
or Generic Probability Theory
(GPT in short, see \cite{barnum2007generalized} and references therein,
although it seems that the idea appeared for the first time
in 1974 in the Mielnik's paper \cite{mielnik1974generalized}).
The basic notion in the GPT is an arbitrary convex set of states.
GPT is more general than the quantum logic framework that we use in this paper.
On the other hand, the latter, due to its more restrictive nature,
provides more tools to study features of the box world theories.
Thus, our work can be considered as complementary
to the convex set framework.

\section{Quantum logics as a framework for generalized probability}

Quantum logic approach originates from 
the seminal paper of Birkhoff and von Neumann \cite{birkhoff1936logic}.
A detailed physical introduction and justification of the whole programme
can be found in the book of Piron \cite{piron1976foundations},
where the Hilbert space formulation of quantum mechanics is derived
from the set of purely logical axioms.
For us, it is important to note, that the notion of \emph{probability}
was always modeled over some physical system 
\footnote{In the famous list of Hilbert problems \cite{hilbert},
  axiomatic treatment of probability was the most important
  task of the 6th problem:
  \emph{Mathematical Treatment of the Axioms of Physics.}}.
In case of classical probability, 
the sample space can be treated as a classical phase space
(set of pure classical states, not neccessarily Hamiltonian system),
random events correspond to subsets of the phase space
and probability measure is classical (in general mixed) state.
So, by analogy, quantum probability can be defined by
specifing the set random events, 
that correspond to orthogonal projectors on the Hilbert space of the system
and a density matrix that defines a measure 
on the set of all projectors 
(we lack the notion of sample space, 
but that is why the quantum probability is quantum).
These ideas motivate definitions presented in this section.

\begin{dfn}
    A \emph{quantum logic} is a partialy ordered set $\mathcal L$
    with a map $\ocmpl\colon\mathcal L\to \mathcal L$ such that
    \begin{enumerate}
        \item[L1] there exists the greatest (denoted by $\bbone$) 
            and the least (denoted by $0$) element in $\mathcal L$,
        \item[L2] map $p \mapsto p\ocmpl$ is order reversing, i.e.\
            $p \le q$ implies that $q\ocmpl \le p\ocmpl$,
        \item[L3] map $p \mapsto p\ocmpl$ is idempotent, i.e.\
            $(p\ocmpl)\ocmpl = p$,
        \item[L4] for a countable family $\set{p_i}$, s.t.\ $p_i \le p_j\ocmpl$
            for $i\neq j$, the supremum $\bigvee \set{p_i}$ exists,
        \item[L5] if $p\le q$ then $q = p \vee(q\wedge p\ocmpl)$ 
            (orthomodular law),
    \end{enumerate}
    where 
    $p\vee q$ is the least upper bound 
    and $p \wedge q$ the greatest lower bound of $p$ and $q$.
\end{dfn}

Elements of a quantum logic $\mathcal L$
are interpreted as \emph{propositions} about a physical system
(equivalence class of experimental setups that
result in one of two outcomes: either \true\ or \false).
The partial order relation $p \le q$ 
is interpreted as ``$q$ is more plausible than $p$''
(e.g.\ whenever $p$ is \true, $q$ is \true\ as well). 
The greatest element of $\mathcal L$ corresponds 
to trivial experimental questions, i.e.\ 
the ones that always result in \true.
Contrary, the least element corresponds to
experimental questions that are always \false.
Map $p\mapsto p\ocmpl$ encodes \emph{negation},
i.e.\ whenever $p$ is \verb|true|, $p\ocmpl$ is \verb|false|, etc.
Then L2 and L3 simply encode the basic properties of negation.
Whenever $p \le q\ocmpl$ we say that $p$ and  $q$ are \emph{disjoint}
and denote it by $p \perp q$.
It is clear, that it means that $p$ and $q$ are mutually exclusive
and consequently it should be meaningful to ask question ``$p$ or $q$''.
L4 extends this intuition to any countable family of disjoint elements.
The last one, L5, lacks direct interpretation.
One can think of it as a very form of weak distributivity.
Nevertheless, it has profound technical importance.

If for any $p, q\in\mathcal L$, $p\vee q, p \wedge q$ exist,
then $\mathcal L$ is said to be an \emph{orthomodular lattice}.
On the other hand, the set of projectors on the Hilbert space,
ordered by subspace inclusion and with $P\ocmpl = I - P$
is always an orthomodular lattice.
Thus we can define quantum probability as a orthomodular lattice.
Moreover, if the distributivity law $p \vee (q \wedge r) = (p \vee q)\wedge(p \vee r)$
holds, then $\mathcal L$ is a Boolean algebra, 
i.e.\ one can find a set $\Omega$ such that $\mathcal L$
can be identified with a $\sigma$-algebra of its subsets.

Finally, let us point that it follows directly from L1 and L2
that both of de Morgan laws are satisfied, i.e.\
if $p\vee q$ exists then $p\ocmpl\wedge q\ocmpl$ exists
and is equal to $(p\vee q)\ocmpl$,
similarly if $p\wedge q$ exists then $p\ocmpl\vee q\ocmpl$ exists
and is equal to $(p\wedge q)\ocmpl$.

\begin{dfn}
  An element $p\neq 0$ of a quantum logic $\mathcal L$
  is called an \emph{atom} whenever $0 \le q \le p$
  implies that $q = 0$ or $q = p$.
  $\mathcal L$ is called \emph{atomic} whenever for any $q\in \mathcal L$
  there exists an atom $p \le q$.
  It is called \emph{atomistic}, whenever any element $q\in \mathcal L$
  is a supremum of atoms less than $q$. 
\end{dfn}

\begin{dfn}
    A subset $\mathcal K\subset\mathcal L$ of quantum logic $\mathcal L$
    is a \emph{sublogic,} whenever 
    (i) if $p\in\mathcal K$ then $p\ocmpl\in\mathcal K$
    (ii) for any countable family $\set{p_i}\subset\mathcal K$ 
    of mutually disjoint elements $\bigvee\set{p_i} \in \mathcal K$.
\end{dfn}

\begin{dfn}
    Let $\mathcal L$ be a quantum logic.
    A \emph{state} $\rho$ on $\mathcal L$ is a map
    $\rho\colon\mathcal L \to [0, 1]$, s.t.\
    \begin{enumerate}
        \item[S1] $\rho(\bbone) = 1$,
        \item[S2] for a countable family $\set{p_i}$, s.t.\ $p_i \le p_j\ocmpl$
            $\rho(\bigvee\set{p_i}) = \sum_i p_i$.
    \end{enumerate}
    Set of all states for quantum logic $\mathcal L$
    will be denoted by $\mathcal S(\mathcal L)$.
\end{dfn}

Value $\rho(p), p\in\mathcal L$ can be interpreted 
as probability of getting answer \true\ for $p$.
It is straightforward to check that this agrees
with the interpretation of elements of $\mathcal L$.
In particular, it can be easily shown, 
that if $p \le q$ then $\rho(p) \le \rho(q), \forall \rho$.
In physical applications, we usually assume that we have
enough states to determine the order, i.e.\
if $\rho(p) \le \rho(q), \forall \rho$ then $p \le q$.

From the physical point of view we are often interested
in subsets of observables which outcomes can be described
by classical probability models
(e.g.\ complete set of commuting observables in quantum physics,
or in other words, simultenously measureable observables).
For this purpose we define:

\begin{dfn}
    Let $\mathcal L$ be a quantum logic.
    Then $p, q\in\mathcal L$ are \emph{compatible},
    what we denote by $p\compat q$,
    whenever there exist pairwise disjoint questions $p_1, q_1, r$
    such that $p = p_1 \vee r, q = q_1 \vee r$.

    More generaly, a subset $A\subset\mathcal L$
    is said to be \emph{compatible} whenever
    for any finite subset $\set{p_1, \dots, p_n}\subset A$
    there exist finite subset $G\subset\mathcal L$,
    such that (i) elements of $G$ are mutually disjoint,
    (ii) any $p_i$ is supremeum of some subset of $G$.
\end{dfn}

This definition looks very technical.
It is easier to think of compatibility in terms of the following property:

\begin{theorem}[\cite{ptak1991orthomodular}, Thm. 1.3.23]
    Let $A\subset\mathcal L$ be a compatible subset of quantum logic.
    Then there exists a Boolean sublogic $\mathcal K\subset \mathcal L$,
    s.t. $A\subset \mathcal K$. 
    \label{thm:compat}
\end{theorem}

Observe, that in general pairwise compatiblity in $A\subset \mathcal L$
is not sufficient for existence of a Boolean sublogic containing $A$
(cf.\ \cite{ptak1991orthomodular} for examples)

Although elements of $\mathcal L$ can be abstract,
we will be interested in one particular class of quantum logics,
namely:

\begin{dfn}[see \cite{ptak1991orthomodular}, Sec. 1.1]
    A \emph{concrete logic} $\Delta$
    is a family of subsets of some set $\Omega$
    with partial order relation given by set inclusion 
    and $A\ocmpl = \Omega\setminus A$ satisfying:
    \begin{enumerate}
        \item[C1] $\emptyset\in\Delta$,
        \item[C2] $A\in \Delta$ implies $\Omega\setminus A \in \Delta$,
        \item[C3] for any countable family $\set{A_i}\subset \Delta$
            of mutually disjoint sets
            $\bigcup \set{A_i} \in \Delta$.
    \end{enumerate} 
\end{dfn}

\begin{example}
    Let $\Omega = \set{1,\dots, 2k}$ 
    and $\Delta$ be a family of subsets $\Omega$
    with even number of elements.
    Then $\Delta$ is a concrete logic,
    which is a Boolean algebra for $k=1$,
    an orthomodular lattice for $k=2$ and
    a quantum logic for $k=3$.
\end{example}

Although concrete logics are obviously more general than
Boolean algebras and even orthomodular lattices,
they exhibit some classical properties.
For example, Heisenberg uncertainty relations
are not satisfied in concrete logics \cite{pulmanova2007}, 
precisely

\begin{dfn}
  An \emph{observable} $X$ on a quantum logic $\mathcal L$
  is a map $X\colon\borel(\reals) \to \mathcal L$,
  where $\borel(\reals)$ is a set of Borel subsets of real line,
  such that:
  \begin{enumerate}
  \item[O1] $X(\reals) = \bbone$,
  \item[O2] $X(\reals\setminus A) = X(A)\ocmpl$,
  \item[O3] $X(A_1\cup \dots \cup A_n) = X(A_1)\vee\dots\vee X(A_n)$
    for any family of mutually disjoint $A_i$'s.
  \end{enumerate}
\end{dfn}

For an observable $X$ we can define its expected value
and variance by
\begin{align*}
    \mu(X) & \defeq \int_\reals t \mu(X(dt)),\\
    \Delta_\mu X & \defeq \int_\reals (t-\mu(X))^2 \mu(X(dt)),
\end{align*}
whenever integrals exist.
Then for any pair of observables $X, Y$
with finite expected value and variance
one of two conditions is satisfied, either
\begin{align}
  &\forall \varepsilon>0\,
  \exists \hbox{ a state }\mu \hbox{ with finite variance for } X \hbox{ and } Y,\nonumber\\
  &(\Delta_\mu X)(\Delta_\mu Y) < \varepsilon
  \label{eq:heisenberg-fail}
\end{align}
or
\begin{align}
  &\exists \varepsilon>0\,
  \forall \hbox{ states } \mu \hbox{ with finite variance for } X \hbox{ and } Y,\nonumber\\
  &(\Delta_\mu X)(\Delta_\mu Y) \ge \varepsilon.
  \label{eq:heisenberg}
\end{align}
In the former case we say that
\emph{Heisenberg uncertainty relations are not satisfied,}
while in the latter case we say that
\emph{Heisenberg uncertainty relations are satisfied.}
Then it follows that if $\mathcal L$ is a concrete logic
then the Heisenberg uncertainty relations are not satisfied
(see Thm.~50 and Thm.~129 in \cite{pulmanova2007}).

\section{Concrete logic of non-signalling theories}

Let us fix our attention on arbitrary
box world $(\mathcal U, \mathcal V)$,
with $N, M$ be the number of distinct input values
on left box and right box respectively.
We define sets
\begin{align*}
    \Gamma_1 &= \set{(x_1, x_2, \dots, x_N) \setdef x_i \in \mathcal U_i},\\
    \Gamma_2 &= \set{(y_1, y_2, \dots, y_M) \setdef y_i \in \mathcal V_i},\\
    \Gamma   &= \Gamma_1\times \Gamma_2
\end{align*}
If boxes obeyed classical physics,
then $\Gamma_1, \Gamma_2, \Gamma$ would correspond
to the phase spaces of left box, right box and the phase
space of a composite system.

With any experimental question of the form:
\begin{equation}
  \parbox{20em}{``does intput $a$ on left box and $b$ on right\\
    result in output $\alpha$ on left box and $\beta$ on right?''}
  \label{eq:expq}
\end{equation}
we assign an element $[a\alpha, b\beta]\subset\Gamma$ 
in the following way:
\begin{equation}
  [a\alpha, b\beta] = \set{(x, y)\in \Gamma \setdef x_a = \alpha, y_b = \beta}
  \label{eq:atoms}
\end{equation}
We denote by $\mathcal A$ the set of all such subsets of $\Gamma$.
Let $\mathcal L$ be a sublogic of Boolean algebra $2^\Gamma$
generated by $\mathcal A$
(the smallest sublogic containing $\mathcal A$;
it exists, see \cite{ptak1991orthomodular}).
It follows that:

\begin{obs}
    $\mathcal L$ is a concrete logic with finite number of elements.
    Moreover, any element of $\mathcal L$ 
    is either a finite union of mutually disjoint sets from $\mathcal A$
    or is an empty set.
\end{obs}

Let us observe that $[a \alpha, b \beta] \perp [c \gamma, d \delta]$
whenever $[a \alpha, b \beta] \cap [c \gamma, d \delta] = \emptyset$,
i.e.\ when (i) $a = c$ and $\alpha \neq \gamma$ or
(ii) $b = d$ and $\beta \neq \delta$.
We will write $p\oplus q$ to indicate $p\vee q$ 
with implicit assumption that $p\perp q$.

\begin{lemma}
    If $p = [a \alpha, b \beta] \leq q$,
    then either 
    (i) $q = p \oplus q'$, 
    (ii) $q = [a \alpha, \bbone] \oplus q'$,
    (iii) $q = [\bbone, b \beta] \oplus q'$ or (iv) $q = \bbone$,
    where
    \begin{align*}
            [a \alpha, \bbone] &= \set{(x, y)\in\Gamma \setdef x_a = \alpha},\\
            [\bbone, b \beta] &= \set{(x, y)\in\Gamma \setdef y_b = \beta}.
    \end{align*}
    \label{thm:order-rel}
\end{lemma}
\begin{proof}
    Let $q = q_1 \oplus \dots \oplus q_n$ with
    $q_i = [a_i \alpha_i, b_i \beta_i]\in\mathcal A$.
    Alternative (i) is obvious, so let us assume that none of $q_i$'s equals $p$.
    Let first $k$ of $q_i$'s be not disjoint with $p$, 
    i.e.\ $p \notperp q_i$ for $i=1,\dots, k$ 
    and $p\perp q_i$ for $i=k+1, \dots, n$.
    Thus any of $q_i$'s for $i=1, \dots, k$ must be of the form:
    (I) $q_i = [a_i \alpha_i, b \beta]$,
    (II) $q_i = [a \alpha, b_i \beta_i]$ or
    (III) $q_i = [a_i\alpha_i, b_i \beta_i]$,
    where $a_i \neq a$, $b_i \neq b$.

    Let us firstly examine the case when one of $q_i$'s,
    say $q_1$, is of the form (I).
    By mutual disjointness of $q_i$'s
    all remaining $i=2,\dots, k$ are either of form (I) or (III).
    Let $\alpha'\in\mathcal U_{a_1}, \alpha' \neq \alpha_1$
    be another output for input $a_1$.
    Then $(x, y)\in p$, where $x_a = \alpha, x_{a_1} = \alpha'$,
    $y_b = \beta$ and rest arbitrary, but $(x, y)\notin q_1$.
    Consequently, there must be another $q_i$,
    say $q_2$, such that $(x, y)\in q_2$.
    By previous comment, either
    (a) $q_2 = [a_1 \alpha', b \beta]$ (form (I) disjoint with $q_1$) or
    (b) $q_2 = [a_1 \alpha', b_2 \beta_2]$ (form (III) disjoint with $q_1$).

    Assume that (b) is the case.
    Take another $(x', y')\in p$ equal to $(x, y)$
    except for $y_{b_2} = \beta' \neq \beta_2$.
    Then $(x', y')\notin q_1$ since $x_{a_1} \neq \alpha_1$
    and $(x', y')\notin q_2$ since $y_{b_2} \neq \beta_2$. 
    Again, one of $q_i$'s must be of the form 
    $q_i = [a_1 \alpha', b_2 \beta']$ (form (III) disjoint with $q_1$ and $q_2$).
    Observe that there is no admissible form (I) 
    disjoint with $q_1$ and $q_2$ for that case.
    If we repeat this reasoning for all outputs $\mathcal V_{b_2}$
    we conclude that there is some set $I$ of $i=1,\dots, k$
    such that $\bigoplus_{i\in I} q_i = [a_1 \alpha', \bbone]$.
    But $[a_1 \alpha', \bbone] = [a_1 \alpha', b \beta]\oplus r$,
    where $r\perp p$, thus case (b) reduces to (a).

    Consequently, we can assume without loss of generality (a) is always the case.
    We repeat reasoning for all outputs in $\mathcal U_{a_1}$
    and conclude that there must be some set $J$ of $i=1, \dots, k$
    such that $\bigoplus_{i\in J} q_i = [\bbone, b \beta]$. 
    The case when one of $q_i$'s is of the form II is symmetric.
    
    Finally, let $q_1 = [a_1 \alpha_1, b_1 \beta_1]$.
    Take $(x, y)\in p$ such that $x_a = \alpha, x_{a_1} = \alpha'$,
    $y_b = \beta, y_{b_1} = \beta_1$ and rest arbitrary.
    Clearly $(x, y)\notin q_1$. Let $(x, y) \in q_2$.
    Since all of $q_i$'s must be of the form III, 
    $q_2 \perp q_1$ result in either
    (a) $q_2 = [a_1 \alpha', b_2 \beta_2], b_2 \neq b_1$ or 
    (b) $q_2 = [a_1 \alpha', b_1 \beta_1]$.

    Assume (a). 
    Take another $(x', y')\in p$ equal to $(x, y)$ except for $y'_{b_2} = \beta'$. 
    Then there must be $q_i$, such that $(x', y')\in q_i$ 
    and $q_i \perp q_1$ and $q_i \perp q_2$ require that
    $q_i = [a_1 \alpha', b_2 \beta']$.
    Like previously, we repeat this for all possible outcomes of $b_2$
    and conclude that there must be subset $I$ of $i=1,\dots, k$ 
    such that 
    $\bigoplus_{i\in I} q_i = [a_1 \alpha_1, \bbone] = [a_1 \alpha_1, b \beta]\oplus r$,
    so we can rewrite $q$ in the way that there is $q_i$ of the form (I).

    On the other hand, if (b) is the case, we repeat
    reasoning for all other outputs of $a_1$ and
    then symmetrically for all outputs of $b_1$
    and conclude that $q = \bbone$.
\end{proof}

\begin{thm}
    Let $P$ be a PR-box state on $(\mathcal U, \mathcal V)$-box world
    and let $\mathcal L$ be an above defined logic.
    Then $\rho_P\colon\mathcal L\to[0, 1]$ defined by 
    \begin{enumerate}[(i)]
        \item $\rho_P(\emptyset) = 0$,
        \item $\rho_P([a \alpha, b \beta]) = P(\alpha \beta| ab)$,
        \item $\rho_P(\bigvee_i p_i) = \sum_i \rho_P(p_i)$, 
            for any set $\set{p_i}\subset\mathcal A$ of pairwise disjoint elements
    \end{enumerate}
    is a state on $\mathcal L$.
    On the other hand, any state $\rho$ on $\mathcal L$
    defines a PR-state $P_\rho$ on $(\mathcal U, \mathcal V)$-box world
    by: $P_\rho(\alpha\beta| ab) = \rho([a\alpha, b\beta])$.
    \label{thm:states}
\end{thm}
\begin{proof} 
    We need to show that the definition of $\rho_P$ is correct,
    i.e.\ for all $p\in\mathcal L$,
    and for all decompositions $p = p_1\oplus \dots \oplus p_n$ 
    (iii) results in the same value.
    Then the fact that $\rho$ is a state follows directly from its definition.
   
    The simplest case $p = \bbone$ follows from normalization of PR-box state. 
    If $p = [a \alpha, \bbone]$ or $p = [\bbone, b \beta]$
    for some $a, \alpha, b, \beta$
    then decomposition is not unique,
    but non-signalling condition guarantess 
    that (iii) does not introduce any ambiguity.
    It follows from lemma \ref{thm:order-rel}
    that $p$ has non-unique decomposition into atoms only if
    $[a \alpha, \bbone] \le p$ or $[\bbone, b \beta] \le p$ 
    for some $a, \alpha, b, \beta$
    Consequently, any $p$ can be decomposed into:
    $p = p_1 \op \dots \op p_k$
    where $p_i$ has one of three forms: 
    (A) $p_i = [a_i \alpha_i, b_i \beta_i]$, 
    (B) $p_i = [a_i \alpha_i, \bbone]$ or
    (C) $p_i = [\bbone, b_i \beta_i]$.
    
    Consequently, it remains to show that (iii) is valid in case of
    $[a \alpha, \bbone] \op p = [\bbone, b \beta] \op q$.
    Let $q = q_1 \op \dots \op q_k$ be decomposition into atoms. 
    Observe that for all $(x, y)\in p$, with $y_b = \beta'\neq \beta$
    $(x, y) \notin [\bbone, b \beta]$, so for any 
    $\beta' \in \mathcal V_b\setminus{\beta}$ there is (possible more than one)
    $q_i = [c_i \gamma_i, b \beta']$.

    Let $q_1 = [c \gamma, b \beta_1]$ and assume that $c_1 \neq a$.
    Take $(x, y)\in [a \alpha, \bbone]$, such that 
    $x_a = \alpha, x_c = \gamma' \neq \gamma, y_b = \beta_1$. 
    Then $(x, y)\notin [\bbone, b \beta]$ and $(x, y)\notin q_1$.
    There must be $q_i = [c_i \gamma_i, b \beta_i]$
    such that $(x, y)\in q_i$, but $q_i \perp q_1$ requires that 
    $c_i = c, \gamma_i = \gamma'$.
    Repeating this reasoning we conclude that there must be subset
    $I$ of $i=1,\dots,k$ such that $\bigoplus_{i\in I} q_i = [\bbone, b \beta_1]$.
    On the other hand, if $c_1 = a$, then we immediately get that
    $q_1 = [a \alpha, b \beta_1]$. 
    Consequently, $q$ can be written in the form
    $q = \bigoplus_{\beta'\in \mathcal V_b^1}[a \alpha, b \beta']
    \op \bigoplus_{\beta'\in\mathcal V_b^2}[\bbone, b \beta']$
    where $\mathcal V_b^1 \cup \mathcal V_b^2 = \mathcal V_b\setminus{\beta}$.
    The same can be done for $p$.
    Now it is easy to observe that by expanding terms
    $[a \alpha_i, \bbone]$ and $[\bbone, b \beta_i]$
    we can obtain the same atomic decomposition on the left and right.
    By transitivity of equality relation,
    we conclude that (iii) is well defined.

    The second part of theorem follows immediately from the fact,
    that non-signaling condition actually means that there
    are elements of the type $[a\alpha, \bbone], [\bbone, b\beta]$:
    \begin{equation*}
      \sum_{\beta\in\mathcal V_b} \rho([a\alpha, b\beta]) = \rho([a\alpha,\bbone]) =
      \sum_{\gamma\in\mathcal V_c}\rho([a\alpha,c \gamma]),\quad\text{etc.}
    \end{equation*}
\end{proof}

Now we are going to argue that $\mathcal L$
is the logic of $(\mathcal U, \mathcal V)$-box world.
Firstly, let us observe that the set of states on $\mathcal L$
is order determining.
One the other hand, in the operational construction of the logic
of the physical system one \emph{assumes} that the set of physical
states determines the order.
Consequently, states of the logic of $(\mathcal U, \mathcal V)$-box world
induced by the PR-states should be order determining.
Secondly, elements of $\mathcal A$ clearly correspond to atoms of
the $(\mathcal U, \mathcal V)$-box world logic.
Moreover, definition of the $(\mathcal U, \mathcal V)$-box world
explicitly enumerates all possible the most elementary experiment
that we can perform on the box world.
This somehow forces us to assume,
that the logic of $(\mathcal U, \mathcal V)$-box world
is not only atomic, but also atomistic.
Summarizing, the logic $\mathcal L$ is the minimal completion
of the set $\mathcal A$ that has sufficiently many states
(all PR-box states are represented)
and is consistent with definition of box world (all states are non-signaling).

Although the definition of $\mathcal L$ as a certain
concrete logic modeled on a classical system seems to be a lucky guess,
it emerged from the study of properties the abstractly constructed
logic of exemplary box world system with binary input and binary output
for both boxes \cite{tylec2015}.

Elements $[a \alpha, \bbone], [\bbone, b \beta]$ (and their valid $\oplus$-sums)
can be interpreted as propositions about only one of boxes,
thus we call them \emph{localized in A} or \emph{localized in B}
respectively.
Observe, that $[a \alpha, \bbone]\perp [a' \alpha', \bbone]$
if and only if $a = a'$ and $\alpha\neq\alpha'$.
We will use following notation:
\begin{equation*}
    [a\in \mathcal P, \bbone] = \oplus\set{[a \alpha, \bbone]\setdef \alpha\in \mathcal P},
    \text{ where } \mathcal P\subset \mathcal U_a.
\end{equation*}

\begin{proposition}
    Any pair of elements $p, q\in \mathcal L$
    such that $p$ is localized in A and $q$ is localized in B
    is compatible.
    \label{thm:compatibility}
\end{proposition}
\begin{proof} 
    Let $p = [a\in\mathcal P, \bbone]$, $q = [\bbone, b\in\mathcal Q]$.
    Let us define 
    $p_1 = \bigoplus_{\alpha\in\mathcal P,\beta\in\mathcal V_b\setminus\mathcal Q} [a \alpha, b \beta]$, 
    $q_1 = \bigoplus_{\alpha\in\mathcal U_a\setminus\mathcal P, \beta \in\mathcal Q} [a \alpha, b \beta]$,
    $r   = \bigoplus_{\alpha\in\mathcal P, \beta\in\mathcal Q} [a \alpha, b \beta]$.
    It is clear that $p_1, q_1, r$ are mutually disjoint and $p = p_1 \oplus r$,
    $q = q_1 \oplus r$.
\end{proof}

\begin{lemma}
    Let $p = [a\in \mathcal P, \bbone], q = [a'\in \mathcal Q, \bbone]$.
    Then:
    \begin{align*}
        p \wedge q &= 
        \begin{cases}
            0 & \text{if } a \neq a',\\
            [a\in \mathcal P\cap \mathcal Q, \bbone] & \text{otherwise,}
        \end{cases}\\
        p \vee q &=
        \begin{cases}
            \bbone & \text{if } a \neq a',\\
            [a\in \mathcal P\cup \mathcal Q, \bbone] & \text{otherwise.}
        \end{cases}
    \end{align*}
    \label{thm:vee-wedge}
\end{lemma}
\begin{proof} 
    Firstly, let us assume that $a\neq a'$.
    Then $p \wedge q$ must be a subset of
    \begin{equation*}
        p\cap q = \set{(x, y)\setdef x_a\in \mathcal P, x_{a'}\in\mathcal Q}
    \end{equation*}
    but this set has constraints on two elements of $x$,
    so there is no non empty $r\in\mathcal L$, 
    such that $r\subset p\cap q$.
    Consequently, $p\wedge q = 0$.
    Then $p\vee q = \bbone$ follows from the de Morgan law. 
    When $a = a'$ the proof is obvious 
    ($p \cap q = p\wedge q$ and $p \cup q = p \vee q$).
\end{proof}

\begin{proposition}
    Elements $p = [a\in \mathcal P, \bbone], q = [a'\in \mathcal Q, \bbone]$ 
    are compatible if and only if $a = a'$. 
    \label{thm:localcomp}
\end{proposition}
\begin{proof} 
    We will show the non-trivial implication by contraposition.
    Assume that $a \neq a'$.
    Since $p\ocmpl = [a\in\mathcal U_a\setminus \mathcal P, \bbone]$
    \begin{equation*}
        p \vee (p\ocmpl \wedge q) = p \vee 0 = p,
    \end{equation*}
    and
    \begin{equation*}
        (p \vee p\ocmpl)\wedge(p\vee q) = \bbone \wedge \bbone = \bbone,
    \end{equation*}
    so $p, p\ocmpl, q$ cannot be contained in Boolean sublogic of $\mathcal L$.
\end{proof}

\begin{thm}
    A set of pairwise compatible localized elements 
    $[a\in\mathcal P_1,\bbone],\dots,[a\in\mathcal P_k, \bbone],
    [\bbone, b\in\mathcal Q_1], \dots, [\bbone, b\in\mathcal Q_l]$
    is compatible.
    \label{thm:compatible-sets}
\end{thm}
\begin{proof} 
    Denote $\mathcal P = \bigcup_{i=1}^k \mathcal P_i$
    and $\mathcal Q = \bigcup_{j=1}^l \mathcal Q_j$.

    Clearly there is a mutually disjoint partition
    $\set{\tilde{\mathcal P_s}}_{s=1}^K$ of $\mathcal P$,
    such that any $\mathcal P_i = \bigcup_{s\in I_i} \tilde{\mathcal P_s}$
    for some $I_i \subset \set{1\dots K}$.
    Similarly, there is a mutually disjoint partition
    $\set{\tilde{\mathcal Q_s}}_{s=1}^L$ of $\mathcal Q$
    such that any $\mathcal Q_j = \bigcup_{s\in J_j}\tilde{\mathcal Q_s}$
    for some $J_j \subset \set{1\dots L}$.
    Let us define:
    \begin{align*}
        p_{i} &= \bigoplus\set{[a \alpha, b \beta]\setdef 
            \alpha\in\tilde{\mathcal P}_{\tilde i},
            \beta \in \mathcal V_b\setminus \mathcal Q},\\
        q_{j} &= \bigoplus\set{[a \alpha, b \beta]\setdef
            \alpha\in\mathcal U_a\setminus \mathcal P,
            \beta\in\tilde{\mathcal Q}_{\tilde j}},\\
        r_{ij} &= \bigoplus\set{[a \alpha, b \beta]
            \setdef \alpha\in\tilde{\mathcal P}_{\tilde i},
            \beta\in\tilde{\mathcal Q}_{\tilde j}}.
    \end{align*}

    Clearly $p_{i}, q_{j}, r_{ij}$ 
    are mutually disjoint.
    Moreover
    \begin{align*}
        [a\in\mathcal P_i, \bbone] &= \bigoplus_{s\in I_i}
        \left(p_s\oplus \bigoplus_{j=1}^l
        \bigoplus_{t\in J_j} r_{st}\right),\\
        [\bbone, b\in\mathcal Q_j] &= \bigoplus_{t\in J_j}
        \left(q_{t}\oplus \bigoplus_{i=1}^k
        \bigoplus_{s \in I_i} r_{st}\right).
    \end{align*}
\end{proof}

The last Theorem is crucial for the interpretation of
$(\mathcal U, \mathcal V)$-box world system as a system composed
of two separate subsystems.
As was noted in the proof of Thm.~\ref{thm:states}
non-signaling condition is statement of merely existence
of certain elements of the logic.
It is far from obvious that this implies any form
of compatibility in general
(although we do not rule out such possibility). 
It is worth to mention here,
that Coecke points out that
in the context of process theory
the non-signaling condition is
also not the most adequate notion \cite{Coecke:2014aa}.

\begin{cor}
  The logic of a single box in an
  $(\mathcal U, \mathcal V)$-box world
  is a \emph{$0,\bbone$-pasting} of Boolean logics
  $\mathcal B_a = 2^{\mathcal U_a}$,
  i.e.\ the logic $L_1$ of left box
  is a disjoint union of $\set{B_a}_{a=1}^N$
  modulo by equivalence relation that
  identify $0$'s and $\bbone$'s of all $B_a$'s
  (cf. Figure~\ref{fig:01-pasting}).
  Consequently, the logic of single box is
  an orthomodular lattice.
\end{cor}
\begin{proof}
  It is clear that
  $\mathcal B_a = \set{[a\in\mathcal P, \bbone]\setdef\mathcal P \subset \mathcal U_a}$
  is a sublogic of $\mathcal L$, which is a Boolean logic.
  Then the claim follows directly from the Lemma~\ref{thm:vee-wedge}.
\end{proof}

\begin{figure}[t]
  \centering
  \begin{tikzpicture}[node distance=0.7cm, auto,]
    \tikzset{
      bool/.style={rectangle, draw=black},
      empty/.style={rectangle, draw=none},
      arr/.style={->, >=latex', shorten >=1pt, shorten <= 1pt}
    };
    \node[bool] (b1) {$2^{\mathcal U_1}$};
    \node[bool, right=of b1] (b2) {$2^{\mathcal U_2}$};
    \node[empty, right=of b2] (dots) {$\dots$};
    \node[bool, right=of dots] (bn) {$2^{\mathcal U_N}$};
    \node[empty, below=of b1, xshift=5mm] (e1) {$\emptyset$};
    \node[empty, below=of b2] (e2) {$\emptyset$};
    \node[empty, below=of bn, xshift=-5mm] (en) {$\emptyset$};
    \node[empty, above=of b1, xshift=5mm] (o1) {$\bbone$};
    \node[empty, above=of b2] (o2) {$\bbone$};
    \node[empty, above=of bn, xshift=-5mm] (on) {$\bbone$};

    \draw[arr] (e1) -- (b1);
    \draw[arr] (e2) -- (b2);
    \draw[arr] (en) -- (bn);

    \draw[arr] (b1) -- (o1);
    \draw[arr] (b2) -- (o2);
    \draw[arr] (bn) -- (on);

    \node[draw, ellipse, dotted, fit=(e1) (e2) (en)] (eg) {};
    \node[draw, ellipse, dotted, fit=(o1) (o2) (on)] (og) {};
  \end{tikzpicture}
  \caption{The $0,\bbone$-pasting of Boolean algebras
    $\mathcal B_a = 2^{\mathcal U_a}$, i.e.\ disjoint sum of
    all Boolean blocks $\mathcal B_a$, with minimal and maximal elements
    identified.}
  \label{fig:01-pasting}
\end{figure}
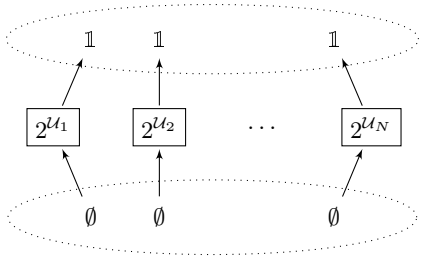

We would like to emphasize that this shows that the logic
of single boxes in box-world models is naively simple.
It consists of the Boolean logics glued together in the most ``free'' way. 
There are no relations imposed between observables defined by different inputs.

\section{Conclusions}

We constructed propositional system for 
a box world model consisting of two boxes
with any finite number of inputs and outputs.
Obtained structure justifies
probabilistic interpretation of box world models
in the sense of quantum logics.
Our construction extends in an obvious way
to many box models, although their properties
were not yet studied by us.

Quantum logic approach allows for a more
refined examination of the model that the convex
set approach (in the sense of \cite{barnum2007generalized}). 
In particular, we were able to discuss notion of compatibility,
Heisenberg uncertainty relations and we identified single box logics.

Our results allow to conclude that,
contrary to the common belief
(cf.~\cite{barnum2006cloning},\cite{barrett2007information})
box world theories cannot be considered as a more general
than quantum theory, even if one restricts to unphysical
finite dimensional Hilbert spaces.
The reason is three-fold:
\begin{enumerate}
\item logic of box-world model has always finite number of elements,
  while the logic of even the simplest quantum model (two level system)
  has an infinite number of elements,
\item logic of single boxes, being a $0,\bbone$-pasting of Boolean algebras,
  is structurally much simpler that the propositional systems of quantum
  mechanics,
\item the logic of box-world model is set-representable and consequently
  such models do not satisfy Heisenberg uncertainty relations.
\end{enumerate}
Although such models permit stronger correlations than possible
in quantum mechanics, we emphasize that this is only one
of many aspects of probability theory. 

\bibliographystyle{apsrev4-1}
\bibliography{concrete-ns}

%merlin.mbs apsrev4-1.bst 2010-07-25 4.21a (PWD, AO, DPC) hacked
%Control: key (0)
%Control: author (72) initials jnrlst
%Control: editor formatted (1) identically to author
%Control: production of article title (-1) disabled
%Control: page (0) single
%Control: year (1) truncated
%Control: production of eprint (0) enabled
\begin{thebibliography}{16}%
\makeatletter
\providecommand \@ifxundefined [1]{%
 \@ifx{#1\undefined}
}%
\providecommand \@ifnum [1]{%
 \ifnum #1\expandafter \@firstoftwo
 \else \expandafter \@secondoftwo
 \fi
}%
\providecommand \@ifx [1]{%
 \ifx #1\expandafter \@firstoftwo
 \else \expandafter \@secondoftwo
 \fi
}%
\providecommand \natexlab [1]{#1}%
\providecommand \enquote  [1]{``#1''}%
\providecommand \bibnamefont  [1]{#1}%
\providecommand \bibfnamefont [1]{#1}%
\providecommand \citenamefont [1]{#1}%
\providecommand \href@noop [0]{\@secondoftwo}%
\providecommand \href [0]{\begingroup \@sanitize@url \@href}%
\providecommand \@href[1]{\@@startlink{#1}\@@href}%
\providecommand \@@href[1]{\endgroup#1\@@endlink}%
\providecommand \@sanitize@url [0]{\catcode `\\12\catcode `\$12\catcode
  `\&12\catcode `\#12\catcode `\^12\catcode `\_12\catcode `\%12\relax}%
\providecommand \@@startlink[1]{}%
\providecommand \@@endlink[0]{}%
\providecommand \url  [0]{\begingroup\@sanitize@url \@url }%
\providecommand \@url [1]{\endgroup\@href {#1}{\urlprefix }}%
\providecommand \urlprefix  [0]{URL }%
\providecommand \Eprint [0]{\href }%
\providecommand \doibase [0]{http://dx.doi.org/}%
\providecommand \selectlanguage [0]{\@gobble}%
\providecommand \bibinfo  [0]{\@secondoftwo}%
\providecommand \bibfield  [0]{\@secondoftwo}%
\providecommand \translation [1]{[#1]}%
\providecommand \BibitemOpen [0]{}%
\providecommand \bibitemStop [0]{}%
\providecommand \bibitemNoStop [0]{.\EOS\space}%
\providecommand \EOS [0]{\spacefactor3000\relax}%
\providecommand \BibitemShut  [1]{\csname bibitem#1\endcsname}%
\let\auto@bib@innerbib\@empty
%</preamble>
\bibitem [{\citenamefont {Popescu}\ and\ \citenamefont
  {Rohrlich}(1994)}]{popescu1994quantum}%
  \BibitemOpen
  \bibfield  {author} {\bibinfo {author} {\bibfnamefont {S.}~\bibnamefont
  {Popescu}}\ and\ \bibinfo {author} {\bibfnamefont {D.}~\bibnamefont
  {Rohrlich}},\ }\href@noop {} {\bibfield  {journal} {\bibinfo  {journal}
  {Foundations of Physics}\ }\textbf {\bibinfo {volume} {24}},\ \bibinfo
  {pages} {379} (\bibinfo {year} {1994})}\BibitemShut {NoStop}%
\bibitem [{\citenamefont {Ac\'in}\ \emph {et~al.}(2006)\citenamefont {Ac\'in},
  \citenamefont {Gisin},\ and\ \citenamefont
  {Masanes}}]{PhysRevLett.97.120405}%
  \BibitemOpen
  \bibfield  {author} {\bibinfo {author} {\bibfnamefont {A.}~\bibnamefont
  {Ac\'in}}, \bibinfo {author} {\bibfnamefont {N.}~\bibnamefont {Gisin}}, \
  and\ \bibinfo {author} {\bibfnamefont {L.}~\bibnamefont {Masanes}},\ }\href
  {\doibase 10.1103/PhysRevLett.97.120405} {\bibfield  {journal} {\bibinfo
  {journal} {Phys. Rev. Lett.}\ }\textbf {\bibinfo {volume} {97}},\ \bibinfo
  {pages} {120405} (\bibinfo {year} {2006})}\BibitemShut {NoStop}%
\bibitem [{\citenamefont {Brassard}\ \emph {et~al.}(2006)\citenamefont
  {Brassard}, \citenamefont {Buhrman}, \citenamefont {Linden}, \citenamefont
  {M\'ethot}, \citenamefont {Tapp},\ and\ \citenamefont
  {Unger}}]{PhysRevLett.96.250401}%
  \BibitemOpen
  \bibfield  {author} {\bibinfo {author} {\bibfnamefont {G.}~\bibnamefont
  {Brassard}}, \bibinfo {author} {\bibfnamefont {H.}~\bibnamefont {Buhrman}},
  \bibinfo {author} {\bibfnamefont {N.}~\bibnamefont {Linden}}, \bibinfo
  {author} {\bibfnamefont {A.}~\bibnamefont {M\'ethot}}, \bibinfo {author}
  {\bibfnamefont {A.}~\bibnamefont {Tapp}}, \ and\ \bibinfo {author}
  {\bibfnamefont {F.}~\bibnamefont {Unger}},\ }\href {\doibase
  10.1103/PhysRevLett.96.250401} {\bibfield  {journal} {\bibinfo  {journal}
  {Phys. Rev. Lett.}\ }\textbf {\bibinfo {volume} {96}},\ \bibinfo {pages}
  {250401} (\bibinfo {year} {2006})}\BibitemShut {NoStop}%
\bibitem [{\citenamefont {Barrett}(2007)}]{barrett2007information}%
  \BibitemOpen
  \bibfield  {author} {\bibinfo {author} {\bibfnamefont {J.}~\bibnamefont
  {Barrett}},\ }\href@noop {} {\bibfield  {journal} {\bibinfo  {journal}
  {Physical Review A}\ }\textbf {\bibinfo {volume} {75}},\ \bibinfo {pages}
  {032304} (\bibinfo {year} {2007})}\BibitemShut {NoStop}%
\bibitem [{\citenamefont {Barrett}\ \emph {et~al.}(2005)\citenamefont
  {Barrett}, \citenamefont {Linden}, \citenamefont {Massar}, \citenamefont
  {Pironio}, \citenamefont {Popescu},\ and\ \citenamefont
  {Roberts}}]{barrett2005nonlocal}%
  \BibitemOpen
  \bibfield  {author} {\bibinfo {author} {\bibfnamefont {J.}~\bibnamefont
  {Barrett}}, \bibinfo {author} {\bibfnamefont {N.}~\bibnamefont {Linden}},
  \bibinfo {author} {\bibfnamefont {S.}~\bibnamefont {Massar}}, \bibinfo
  {author} {\bibfnamefont {S.}~\bibnamefont {Pironio}}, \bibinfo {author}
  {\bibfnamefont {S.}~\bibnamefont {Popescu}}, \ and\ \bibinfo {author}
  {\bibfnamefont {D.}~\bibnamefont {Roberts}},\ }\href@noop {} {\bibfield
  {journal} {\bibinfo  {journal} {Physical Review A}\ }\textbf {\bibinfo
  {volume} {71}},\ \bibinfo {pages} {022101} (\bibinfo {year}
  {2005})}\BibitemShut {NoStop}%
\bibitem [{\citenamefont {Pt{\'a}k}\ and\ \citenamefont
  {Pulmannov{\'a}}(2007)}]{pulmanova2007}%
  \BibitemOpen
  \bibfield  {author} {\bibinfo {author} {\bibfnamefont {P.}~\bibnamefont
  {Pt{\'a}k}}\ and\ \bibinfo {author} {\bibfnamefont {S.}~\bibnamefont
  {Pulmannov{\'a}}},\ }in\ \href@noop {} {\emph {\bibinfo {booktitle} {Handbook
  of Quantum Logic and Quantum Structures}}}\ (\bibinfo  {publisher}
  {Elsevier},\ \bibinfo {year} {2007})\ pp.\ \bibinfo {pages}
  {215--284}\BibitemShut {NoStop}%
\bibitem [{\citenamefont {Barnum}\ \emph {et~al.}(2007)\citenamefont {Barnum},
  \citenamefont {Barrett}, \citenamefont {Leifer},\ and\ \citenamefont
  {Wilce}}]{barnum2007generalized}%
  \BibitemOpen
  \bibfield  {author} {\bibinfo {author} {\bibfnamefont {H.}~\bibnamefont
  {Barnum}}, \bibinfo {author} {\bibfnamefont {J.}~\bibnamefont {Barrett}},
  \bibinfo {author} {\bibfnamefont {M.}~\bibnamefont {Leifer}}, \ and\ \bibinfo
  {author} {\bibfnamefont {A.}~\bibnamefont {Wilce}},\ }\href@noop {}
  {\bibfield  {journal} {\bibinfo  {journal} {Physical review letters}\
  }\textbf {\bibinfo {volume} {99}},\ \bibinfo {pages} {240501} (\bibinfo
  {year} {2007})}\BibitemShut {NoStop}%
\bibitem [{\citenamefont {Mielnik}(1974)}]{mielnik1974generalized}%
  \BibitemOpen
  \bibfield  {author} {\bibinfo {author} {\bibfnamefont {B.}~\bibnamefont
  {Mielnik}},\ }\href@noop {} {\bibfield  {journal} {\bibinfo  {journal}
  {Communications in Mathematical Physics}\ }\textbf {\bibinfo {volume} {37}},\
  \bibinfo {pages} {221} (\bibinfo {year} {1974})}\BibitemShut {NoStop}%
\bibitem [{\citenamefont {Birkhoff}\ and\ \citenamefont
  {Von~Neumann}(1936)}]{birkhoff1936logic}%
  \BibitemOpen
  \bibfield  {author} {\bibinfo {author} {\bibfnamefont {G.}~\bibnamefont
  {Birkhoff}}\ and\ \bibinfo {author} {\bibfnamefont {J.}~\bibnamefont
  {Von~Neumann}},\ }\href@noop {} {\bibfield  {journal} {\bibinfo  {journal}
  {Annals of mathematics}\ }\textbf {\bibinfo {volume} {37}},\ \bibinfo {pages}
  {823} (\bibinfo {year} {1936})}\BibitemShut {NoStop}%
\bibitem [{\citenamefont {Piron}(1976)}]{piron1976foundations}%
  \BibitemOpen
  \bibfield  {author} {\bibinfo {author} {\bibfnamefont {C.}~\bibnamefont
  {Piron}},\ }\href@noop {} {\emph {\bibinfo {title} {Foundations of quantum
  physics}}}\ (\bibinfo  {publisher} {WA Benjamin, Inc., Reading, MA},\
  \bibinfo {year} {1976})\BibitemShut {NoStop}%
\bibitem [{Note1()}]{Note1}%
  \BibitemOpen
  \bibinfo {note} {In the famous list of Hilbert problems \cite {hilbert},
  axiomatic treatment of probability was the most important task of the 6th
  problem: \protect \emph {Mathematical Treatment of the Axioms of
  Physics.}}\BibitemShut {Stop}%
\bibitem [{\citenamefont {Pt{\'a}k}\ and\ \citenamefont
  {Pulmannov{\'a}}(1991)}]{ptak1991orthomodular}%
  \BibitemOpen
  \bibfield  {author} {\bibinfo {author} {\bibfnamefont {P.}~\bibnamefont
  {Pt{\'a}k}}\ and\ \bibinfo {author} {\bibfnamefont {S.}~\bibnamefont
  {Pulmannov{\'a}}},\ }\href@noop {} {\emph {\bibinfo {title} {Orthomodular
  Structures as Quantum Logics: Intrinsic Properties, State Space and
  Probabilistic Topics}}},\ Vol.~\bibinfo {volume} {44}\ (\bibinfo  {publisher}
  {Springer},\ \bibinfo {year} {1991})\BibitemShut {NoStop}%
\bibitem [{\citenamefont {Tylec}\ and\ \citenamefont
  {Ku{\'s}}(2015)}]{tylec2015}%
  \BibitemOpen
  \bibfield  {author} {\bibinfo {author} {\bibfnamefont {T.}~\bibnamefont
  {Tylec}}\ and\ \bibinfo {author} {\bibfnamefont {M.}~\bibnamefont
  {Ku{\'s}}},\ }\href@noop {} {\bibfield  {journal} {\bibinfo  {journal}
  {Journal of Physics A: Mathematical and Theoretical}\ }\textbf {\bibinfo
  {volume} {48}},\ \bibinfo {pages} {505303} (\bibinfo {year}
  {2015})}\BibitemShut {NoStop}%
\bibitem [{\citenamefont {Coecke}(2014)}]{Coecke:2014aa}%
  \BibitemOpen
  \bibfield  {author} {\bibinfo {author} {\bibfnamefont {B.}~\bibnamefont
  {Coecke}},\ }\href {http://arxiv.org/abs/1405.3681} {\bibfield  {journal}
  {\bibinfo  {journal} {EPTCS}\ }\textbf {\bibinfo {volume} {172}} (\bibinfo
  {year} {2014})},\ \Eprint {http://arxiv.org/abs/1405.3681} {1405.3681}
  \BibitemShut {NoStop}%
\bibitem [{\citenamefont {Barnum}\ \emph {et~al.}(2006)\citenamefont {Barnum},
  \citenamefont {Barrett}, \citenamefont {Leifer},\ and\ \citenamefont
  {Wilce}}]{barnum2006cloning}%
  \BibitemOpen
  \bibfield  {author} {\bibinfo {author} {\bibfnamefont {H.}~\bibnamefont
  {Barnum}}, \bibinfo {author} {\bibfnamefont {J.}~\bibnamefont {Barrett}},
  \bibinfo {author} {\bibfnamefont {M.}~\bibnamefont {Leifer}}, \ and\ \bibinfo
  {author} {\bibfnamefont {A.}~\bibnamefont {Wilce}},\ }\href@noop {}
  {\bibfield  {journal} {\bibinfo  {journal} {arXiv preprint quant-ph/0611295}\
  } (\bibinfo {year} {2006})}\BibitemShut {NoStop}%
\bibitem [{\citenamefont {Hilbert}(1902)}]{hilbert}%
  \BibitemOpen
  \bibfield  {author} {\bibinfo {author} {\bibfnamefont {D.}~\bibnamefont
  {Hilbert}},\ }\href@noop {} {\bibfield  {journal} {\bibinfo  {journal} {Bull.
  Amer. Math. Soc.}\ }\textbf {\bibinfo {volume} {8}},\ \bibinfo {pages} {437}
  (\bibinfo {year} {1902})}\BibitemShut {NoStop}%
\end{thebibliography}%

\end{document}